%% file: main.tex
\documentclass[11pt]{article}

\usepackage[margin=1in]{geometry} 
\input{packages}
\input{macros}
\input{macros-syntax-coloring}

\title{A Compositional Language for Property Graphs}

\author{
Marcelo Arenas$^{1,2}$ \and
Leonid Libkin$^{1,3,4}$ \and
Wim Martens$^{1,5}$
}

\date{} 

\newcommand{\OMIT}[1]{}
\usepackage{amsthm}

\theoremstyle{plain}
\newtheorem{theorem}{Theorem}
\newtheorem{proposition}[theorem]{Proposition}

\theoremstyle{definition}
\newtheorem{definition}[theorem]{Definition}
\newtheorem{example}[theorem]{Example}

\theoremstyle{remark}
\newtheorem{remark}[theorem]{Remark}

\begin{document}

\maketitle

\begin{center}
{\small
$^1$ RelationalAI \\
$^2$ Pontificia Universidad Cat\'olica de Chile \\
$^3$ University of Edinburgh \\
$^4$ IRIF, Universit\'e Paris-Cit\'e \\
$^5$ University of Bayreuth
}
\end{center}

\begin{abstract}
\input{abstract}
\end{abstract}

\noindent\textbf{Keywords:}
Query languages, Property graphs, Datalog, Compositionality

\parindent=0cm
\parskip=0.27cm

\section{Introduction}
\label{sec-intro}
\input{introduction}

\section{Preliminaries}
\label{sec-prelim}
\input{preliminaries}

\section{Regular Path Queries with Variables}\label{sec:rpqvs}

\input{rpqv}

\section{The Query Language \hashd}
\label{sec-hash-datalog}

\input{hash-datalog}

\section{A Syntax Proposal for GQL and SQL}

\input{gql}

\section{Conclusion}
With RPQVs and $\hashd$ we have designed two independent mechanisms that, if adopted by the GQL and SQL/PGQ standards, will fill their known expressivity gaps. 
Either one separately solves the gap known as the ``increasing values on edges'' query. This is not a randomly chosen query: it is of such importance to the standards committee 
that an entirely new mechanism of post-processing paths with a sliding window was proposed to express it~\cite{tobias,fred}.


RPQVs and $\hashd$, however, both show how the problem can be addressed using mechanisms that are close to those that already exist in the standards. RPQVs can express the query by adopting a fully symmetric treatment of nodes and edges in the design of path pattern expressions. The rule-based system of $\hashd$ can express it since it can transform a graph to its dual, on which the existing GQL mechanism for path matching can express the query.

Combining RPQVs and $\hashd$ allows for complete compositionality: a free flow of information back and forth between graph querying and relational processing in GQL. This is in particular manifested by the capture of all NLOGSPACE queries, which is currently only possible with a significant complexity cost. Our concrete proposal to ISO includes several fallback options to help our main ideas get across and improve this situation for the standards.


Since property graphs are now a major representation model for knowledge graphs, closing the compositionality gaps in GQL and SQL/PGQ is also a step toward more principled, expressive, and interoperable graph querying for the Semantic Web.





\bibliographystyle{plain}
\bibliography{references}


\end{document}

%% file: packages.tex
\usepackage[T1]{fontenc}
\usepackage{graphicx}
\usepackage{url}
\usepackage{xcolor}
\usepackage{todonotes}
\usepackage{amsmath}
\usepackage{listings}
\usepackage{tikz}
\usepackage{latexsym,stmaryrd}
\usepackage{amsfonts}
\usepackage{xspace}

\usetikzlibrary{automata}
\usetikzlibrary{chains}
\usetikzlibrary{arrows}
\usetikzlibrary{calc}
\usetikzlibrary{arrows,positioning,backgrounds} 
\usetikzlibrary{fit}

\tikzset{
    rt/.style={
		rectangle,
		fill = white,
		draw=black, 
		text centered,
		inner sep=0.5ex
		},
    rtt/.style={ 
    	rt,
    	inner sep=0.1ex
    	},
    ert/.style={ 
     	rt,
     	dashed
     	}, 
    ertt/.style={ 
        rtt,
        dashed
        }, 
    rect/.style={ 
        rectangle,
        fill = white,
        rounded corners,
        draw=black, 
        text centered,
        inner sep=0.8ex
        },
    rectw/.style={
        rect,
        draw=white
        },
    erect/.style={ 
    	rect,
    	dashed
    	},
    erectw/.style={ 
     	rectw,
     	dashed
     	},
    arrout/.style={
           ->,
           -latex,
           },
    arrin/.style={
           <-,
           latex-,
           },
    arrb/.style={
           <->,
           >=latex,
           }
}

%% file: macros.tex
\newcommand{\mcomment}[2]{}
\newcommand{\footcomment}[2]{}
\newcommand{\margincomment}[2]{}

\renewcommand{\mcomment}[2]{\ifmmode\margincomment{#1}{#2}\else\footcomment{#1}{#2}\fi}
\renewcommand{\footcomment}[2]{{\color{blue}\textbf{(#1)}}\footnote{\textbf{#1:} #2}}
\renewcommand{\margincomment}[2]{{\color{blue}\textbf{(#1)}}\footnotemark\marginnote{\tiny\textsuperscript{\thefootnote}\textbf{#1:} #2}}

    \newcommand{\wim}[1]{\todo[inline,color=green!30]{Wim: #1}}

\newcommand{\Label}{\mathit{Lab}}
\newcommand{\Value}{\mathit{Val}}
\newcommand{\Key}{\mathit{Key}}
\newcommand{\src}{\mathit{src}}
\newcommand{\tgt}{\mathit{tgt}}
\newcommand{\lab}{\mathit{lab}}
\newcommand{\prop}{\mathit{prop}}

\newcommand{\Fin}{\mathit{Fin}}
\newcommand{\hashd}{\#\textsc{Datalog}}

\newcommand{\Var}{\mathit{Var}}

\newcommand{\cond}{\mathit{cond}}
\newcommand{\regex}{\mathit{rgx}}
\newcommand{\expr}{\mathit{expr}}
\newcommand{\node}{\mathit{node}}
\newcommand{\edge}{\mathit{edge}}
\newcommand{\concat}{\mathit{concat}}
\newcommand{\el}{\mathit{el}}

\newcommand{\Paths}{\mathit{Paths}}
\newcommand{\List}{\mathit{List}}

\newcommand{\last}{\mathit{last}}
\newcommand{\first}{\mathit{first}}

\newcommand{\sem}[1]{\llbracket{#1}\rrbracket}

\newcommand{\dom}{\textit{dom}}

\newlength\boxwidth
\newlength\questionwidth

\newcommand{\gqlcode}[1]{\texttt{#1}}

%% file: macros-syntax-coloring.tex
\usepackage{color}
\definecolor{gray}{rgb}{0.4,0.4,0.4}
\definecolor{darkblue}{rgb}{0.0,0.0,0.6}
\definecolor{darkred}{rgb}{0.45,0,0}
\definecolor{darkgreen}{rgb}{0,0.30,0.20}
\definecolor{darkpurple}{RGB}{120, 0, 180}
\colorlet{keywordcolor}{darkblue}
\newcommand{\kwfont}{\color{keywordcolor}\bfseries}

\colorlet{labelcolor}{darkgreen}
\newcommand{\lblfont}{\color{labelcolor}}

\colorlet{keycolor}{darkred}
\newcommand{\keyfont}{\color{keycolor}}

\colorlet{structcolor}{black}
\newcommand{\structfont}{\color{structcolor}\bfseries}

\lstdefinelanguage{vmgql}
{
  columns=fullflexible,
  otherkeywords=,
  keywordstyle=,
  keywords=[1]{},
  keywordstyle=[1]\color{darkpurple},
  keywords=[2]{ACYCLIC, SIMPLE, TRAIL, WALK, ALL, ANY, SHORTEST, MATCH, PROJ, USE, FILTER, FOR, KEEP, WHERE, OPTIONAL, RETURN, GROUP, GROUPS, LET, INSERT, REMOVE, SET, UNION, INTERSECT, OTHERWISE, EXCEPT, IS, NOT, AND, THEN, OR, YIELD, DISTINCT, AS, DIFFERENT, EDGES, PATHS, CALL,CREATE,NODE,EDGE,FROM,ROW,COLLAPSE},
  keywordstyle=[2]\kwfont,
  keywords=[3]{},
  keywordstyle=[3]\structfont\color{gray},
  keywords=[5]{YachtClub, Person, Account},
  keywordstyle=[5]\lblfont,
  keywords=[6]{Transfer, Owns, Member},
  keywordstyle=[6]\lblfont,
  keywords=[7]{address, fullname, name, id, genre, owner, isBlocked, amount},
  keywordstyle=[7]\keyfont,
  string=[m]{"},
  stringstyle=,
  inputencoding=utf8,
  basicstyle=\small\ttfamily,  
  keepspaces=true,
  showstringspaces=false,
  upquote=true,
  aboveskip=3pt,
  belowskip=3pt,
  numberstyle=\scriptsize,
  backgroundcolor=\color{white},
  numbersep=8pt,
  mathescape,
  comment=[l]{//}
}

\newcommand{\forcedhfill}{\hspace*{0pt plus 1fill}}
\makeatletter
\newcommand{\thequery}{(\arabic{equation})}
\newcounter{nextquery}

\newlength{\lstskip}
\newcommand{\vm@cgql@start}[2][]{%
  \leavevmode\unskip\pagebreak[1]\vspace{\lstskip}\par\noindent%
  \hspace*{\parindent}
  #1%
  \lst@boxtrue%
  \lstset{language=vmgql,boxpos=b,resetmargins=true,mathescape=true,#2}
  }
\newcommand{\vm@cgql@end}[1][]{\forcedhfill#1\par\addvspace{\lstskip}}
\lstnewenvironment{gql}[1][]%
  {\vm@cgql@start[\refstepcounter{equation}]{#1}}%
  {\vm@cgql@end[~\raisebox{1ex}{\thequery}]}
\lstnewenvironment{gql*}[1][]%
  {\vm@cgql@start{#1}}
  {\vm@cgql@end}
\makeatother

\newcommand{\kw}[1]{\mathgql{\kwfont{#1}}\xspace}
\newcommand{\mathgql}[1]{\mbox{\ttfamily#1}}

\DeclareRobustCommand*{\gqlinline}{%
  \ifmmode
    \let\SavedBGroup\bgroup
    \def\bgroup{%
      \let\bgroup\SavedBGroup
      \hbox\bgroup
    }%
  \fi
  \lstinline[language=vmgql]
}

%% file: abstract.tex
A major shortcoming of the recently standardized graph query languages
GQL and SQL/PGQ is their lack of compositionality. Given the
importance of these languages in querying knowledge graphs, we address
this shortcoming and propose both theoretical solutions and a path to
adding them to the new standards.  The highlight of the
non-compositionality problem is that while both GQL and SQL/PGQ can
express graph reachability and all first-order queries, they fall
short of the problems in NLOGSPACE.  In view of the completeness of
reachability for NLOGSPACE under first-order reductions, this is
extremely counterintuitive.  The issue is well recognized by the
standards committee that has been searching for language extensions
to fill the gaps at the level of some specific inexpressible
queries.

We address the issue in a systematic way and propose a language that
fills expressivity gaps by allowing full compositionality between
graph patterns and relational queries. It does so by using two key
components: a cleaned up definition of
regular path queries with variables and data value comparisons,
and a fully compositional graph-to-graph language $\hashd$ with
complete support for constructing new graph elements from nodes,
edges, lists of nodes and edges, and even entire paths.
We show that the resulting language addresses the issues facing the
standards committee, and propose a concrete addition to GQL and SQL/PGQ that
incorporates its main features.

\OMIT{Lack of compositionality is a major shortcoming of current query languages for property graphs such as GQL and SQL/PGQ. Even though these languages can express graph reachability all first-order queries, they cannot express all problems in NLOGSPACE. This is puzzling because graph reachability is NLOGSPACE-complete under first-order reductions. The issue is so deep that the ISO committee is searching for language extensions that will allow expressing important but currently inexpressible queries such as returning all paths whose values increase along edges.

We propose a compositional language that fixes this issue. Building on a cleaned up definition of regular path queries with variables and data value comparisons, $\hashd$ is a fully compositional graph-to-graph language with complete support for constructing new graph elements from nodes, edges, lists of nodes and edges, and even entire paths. \wim{It may be the first such language?}
We prove that $\hashd$ can express all queries in NLOGSPACE and provide a proposal for incorporating its main principles in the GQL and SQL/PGQ standards.

}

%% file: introduction.tex
 Knowledge graphs are an essential tool used by most major enterprises for organizing, integrating, and reasoning over complex data. Their widespread use led to a robust market that is expected to grow tenfold in the next decade. Most knowledge graphs utilise one of the basic underlying models: \emph{RDF graphs} or {\em labeled  property graphs (LPG)}. Several recent market analyses portray the split between these two models as about equal, or perhaps with a slight 60-40 edge to LPGs.\footnote{ \url{www.fortunebusinessinsights.com/knowledge-graph-market-112139} and 
\url{www.marketreportsworld.com/market-reports/graph-database-market-14723746}}
 The decision which model to use hinges upon a particular application area. If fast traversals and analytics are needed (common in, e.g., finance, supply chain, logistics, and cyber security applications), then LPGs are usually chosen. On the other hand, if  reasoning and ontologies become crucial (common in healthcare, government, and linked data), the RDF model is often preferred.

\paragraph{The State of Query Languages for LPGs.} For the RDF model, we have the SPARQL query language, standardized in 2008. Its definition has been through several revisions, is well accepted by industrial stakeholders, and has been subject of a significant research effort on understanding its semantics, expressiveness, and various extensions
\cite{PAG06,P07,sparql,w3c-sparql11-2010,nre,yottabyte,HP16,KaminskiKG17,BarceloKPS-tods18}. The academic community played a crucial role in its design, pointing out design flaws \cite{PAG06,yottabyte,LosemannM-tods13} that have been fixed in the language standard. 

When it comes to LPGs, the situation is quite different. LPG query languages were first proposed around a decade and a half ago, with Cypher \cite{cypher} paving the way, and PGQL \cite{pgql} and GSQL \cite{tigergraph} following. As these were vendor-specific,  the graph database community decided to create a new standard language, similar to how SQL unified the relational database industry in the 1980s. The task was delegated to the same committee that is responsible for SQL. This committee created two languages: an extension of SQL for LPG querying called \emph{SQL/PGQ}; and a standalone language for LPGs called \emph{GQL}. The languages have very significant overlap: their pattern matching facility, the workhorse of every graph query language, is identical \cite{sigmod22}, and their relational operators have the same expressiveness \cite{vldb25}.  Thus, our remarks about GQL directly apply to SQL/PGQ.

For GQL and SQL/PGQ, we are in a similar position than SPARQL 15 years ago. Although their design is strongly inspired by Cypher and 
input from the research community~\cite{gcore},  final design decisions were taken by a committee at a time when the foundational understanding of graph query languages was nowhere near as strong as it was for relational databases in the early days of SQL design.  
As a result, languages for LPGs are only now becoming a subject of serious academic investigation, having been transcribed from the rather idiosyncratic language of standards for the research community \cite{gpc,icdt23,vldb25}.



\paragraph{Major Shortcoming of GQL and SQL/PGQ: Non-Compositionality.} This recent academic investigation uncovered several significant shortcomings of the new languages, the most important of them broadly described as \emph{non-compositionality}. 
The {\em main goal} of this paper is to \emph{address this fundamental limitation}. We first develop a fully compositional theoretical language that captures the essence of GQL and avoids this design flaw. Second, we propose a concrete extension to GQL that addresses these limitations. The language is based on two key components: an extension of regular path queries with variables and flexible path concatenation, and an extension of Datalog with value invention. 

Next, we explain how non-compositionality manifests itself, using an example from \cite{vldb25,pods25}. It concerns two very similar queries, both commonly occurring in practice, and yet one of them not definable in GQL. Recall that in the LPG model, we deal with graphs whose nodes and edges can have labels (or types), and multiple properties attached to them. For example, in a bank database, we can have account nodes and transfer edges. \emph{Account} and \emph{Transfer} are labels; account nodes have properties such as owner and balance, while transfer edges have have properties such as amount and timestamp (ts). 

\resizebox{.9\linewidth}{!}{%
    \begin{tikzpicture}
        \tikzstyle{every state}=[draw,thick,rectangle,rounded corners,fill=white!85!black,minimum size=5mm, text=black, font=\ttfamily, inner sep=0pt]
        \tikzstyle{every node}=[font=\ttfamily]

        \node[state] (acc1) at (-6,0) {
          \begin{tabular}{l}
            \small {\bf owner:} Megan\\
            \small {\bf balance:} 100\\
          \end{tabular}
        };

        \node[state] (acc2) at (0,0) {
          \begin{tabular}{l}
            \small {\bf owner:} Jay\\
            \small {\bf balance:} 150\\
          \end{tabular}
        };

        \node[state] (acc3) at (6,0) {
          \begin{tabular}{l}
            \small {\bf owner:} Robben\\
            \small {\bf balance:} 160\\
          \end{tabular}
        };

        \node[state] (acc4) at (6,-2.5) {
          \begin{tabular}{l}
            \small {\bf owner:} Rebecca\\
            \small {\bf balance:} 200\\
          \end{tabular}
        };

        \node[state] (acc5) at (0,-2.5) {
          \begin{tabular}{l}
            \small {\bf owner:} Scott\\
            \small {\bf balance:} 150\\
          \end{tabular}
        };

        \node[state] (acc6) at (-6,-2.5) {
          \begin{tabular}{l}
            \small {\bf owner:} Mike\\
            \small {\bf balance:} 180\\
          \end{tabular}
        };

        \node[rectangle,draw,fill=white] (acc1label) at ($(acc1)+(0.5,0.6)$) {\ttfamily \small Account};
        \node[rectangle,draw,fill=white] (acc2label) at ($(acc2)+(0.5,0.6)$) {\ttfamily \small Account};
        \node[rectangle,draw,fill=white] (acc3label) at ($(acc3)+(0.6,0.6)$) {\ttfamily \small Account};
        \node[rectangle,draw,fill=white] (acc4label) at ($(acc4)+(1.2,0.6)$) {\ttfamily \small Account};
        \node[rectangle,draw,fill=white] (acc5label) at ($(acc5)+(0.0,0.6)$) {\ttfamily \small Account};
        \node[rectangle,draw,fill=white] (acc6label) at ($(acc6)+(0.0,0.6)$) {\ttfamily \small Account};
      

        \node[state,dashed,fill=white!95!black] (t1) at (-3,0) {
            \begin{tabular}{l}
                \small {\bf ts:} 20260101\\
                \small {\bf amount:} 10\\
            \end{tabular}
        };
        \node[state,dashed,fill=white!95!black] (t2) at (3,0) {
            \begin{tabular}{l}
                \small {\bf ts:} 20260102\\
                \small {\bf amount:} 10\\
            \end{tabular}
        };
        \node[state,dashed,fill=white!95!black] (t3) at (7.5,-1.15) {
            \begin{tabular}{l}
                \small {\bf ts:} 20260103\\
                \small {\bf amount:} 10\\
            \end{tabular}
        };
        \node[state,dashed,fill=white!95!black] (t4) at (3,-2.5) {
            \begin{tabular}{l}
                \small {\bf ts:} 20260104\\
                \small {\bf amount:} 10\\
            \end{tabular}
        };
        \node[state,dashed,fill=white!95!black] (t5) at (-3,-2.5) {
            \begin{tabular}{l}
                \small {\bf ts:} 20260105\\
                \small {\bf amount:} 10\\
            \end{tabular}
        };

        \node[rectangle,draw,dashed,fill=white] (t1label) at ($(t1)+(0.5,0.6)$) {\ttfamily \small Transfer};
        \node[rectangle,draw,dashed,fill=white] (t2label) at ($(t2)+(0.5,0.6)$) {\ttfamily \small Transfer};
        \node[rectangle,draw,dashed,fill=white] (t3label) at ($(t3)+(0.5,0.6)$) {\ttfamily \small Transfer};
        \node[rectangle,draw,dashed,fill=white] (t4label) at ($(t4)+(0.5,0.6)$) {\ttfamily \small Transfer};
        \node[rectangle,draw,dashed,fill=white] (t5label) at ($(t5)+(0.5,0.6)$) {\ttfamily \small Transfer};


    \begin{pgfonlayer}{background}        
        \path[-latex,thick]
        (acc1) edge (acc2)
        (acc2) edge (acc3)
        (acc3) edge[bend left] (acc4)
        (acc4) edge (acc5)
        (acc5) edge (acc6)
        ;
          
    \end{pgfonlayer}

    



    \end{tikzpicture}
    }

\noindent With this example, consider the query: \emph{``Return two accounts such that there is a transfer chain between them in which account balances increase.''} GQL makes this easy. For simplicity in queries below we omit labels, assuming our graph only has account nodes and transfer edges: 
\begin{gql*}
MATCH (x) ((a1) -> (a2) WHERE a1.balance < a2.balance)+ (y) RETURN x,y
\end{gql*}
making use of the Kleene plus (\texttt{+}) to repeat the subpattern one or more times. 

Next consider a very similar query that might be used in fraud analysis:
\emph{``Return two accounts such that there is a transfer chain between them in which timestamps of transfers increase.''} A natural attempt to write it as
\begin{gql*}
MATCH (x) ( -[t1]-> -[t2]-> WHERE t1.ts < t2.ts)+ (y) RETURN x,y
\end{gql*}
does not work: if we have a path with 
consecutive timestamps 11, 12, 1, 2, it will be matched. 
One problem is deeply rooted in the semantics of GQL, particularly path concatenation. In a repetition like Kleene star or plus, paths are concatenated by \emph{merging their end-nodes}. The last pattern is equivalent to 
\begin{gql*}
MATCH (x) ( () -[t1]-> -[t2]-> () WHERE t1.ts < t2.ts)+ (y)
\end{gql*}
with two anonymous nodes inserted, and thus in the new iteration of a repetition, information about the edges of the previous iteration is lost. 

But this example reveals something more fundamental about GQL. 
Both patterns we use 
are variations of reachability queries with some simple additional checks. Since reachability is complete for the class NLOGSPACE, these queries are also in NLOGSPACE, yet only one of them is expressible in GQL. 
This is very puzzling, and to understand why we need to take a birdeye's view of a single transaction in GQL and SQL/PGQ. In these languages:
\begin{enumerate}
    \item \verb+MATCH+ operates on a graph and produces a relation whose columns are variables mentioned in the pattern;
\item Subsequent operations modify this table in a language that has at least the power of relational algebra/first-order logic;
\item In GQL, it is possible that another \verb+MATCH+ operation occurs; it is performed on the original graph, and its result is joined with the current table.
\end{enumerate}
The key differences between GQL and SQL/PGQ are in step 2, where relational languages differ in their presentation (GQL follows Cypher's iterative style; SQL/PGQ uses SQL). 
With this in mind, observe the following:
\begin{itemize}
    \item Every pattern language expresses the graph reachability problem;
    \item GQL and SQL/PGQ can express all first-order queries; and
    \item graph reachability is NLOGSPACE-complete under first-order reductions. 
\end{itemize}
From this, one would naturally conclude that 
GQL should be able to express all properties in NLOGSPACE. However, we have just seen an NLOGSPACE property not expressible in GQL, and in fact there are even DLOGSPACE  properties that basic GQL and SQL/PGQ cannot express \cite{vldb25}. 

To explain why this is the case, we revisit \emph{compositionality}. The completeness of reachability states that we can take any NLOGSPACE problem, turn it into a graph by first-order transformations, and then answer it by computing the reachability relation on that graph. This requires information to flow \emph{back and forth} between the relational and graph components of the language, and our bird's eye view of GQL and SQL/PGQ shows that there is no such facility.

To achieve compositionality, languages such as GQL would need to construct graphs from the results of relational operators and then run pattern matching on those graphs. This feedback loop --- from relations to graphs --- is  currently missing
and leads to limitations in expressiveness in GQL and SQL/PGQ.

Of course full-featured languages have workarounds (SQL is Turing-complete after all), but they carry an enormous complexity price as they require building exponentially many paths and using exponential time algorithms to solve tractable problems, resulting in extremely poor performance \cite{vldb25,sigmod25}.

\paragraph{Our Contributions.}


Our goal is to offer {\em language design guidance} that closes the expressivity gap without increasing complexity. We show that this can be done in a simple and elegant way. The key ideas behind our approach are a clean model of path queries that significantly simplifies the typing of variables compared to GQL \cite[Fig.~2]{gpc}, 
and an extension of the relational querying component of graph languages that enables the results of relational operations to modify the graph.
This can be done by Datalog-style rules with Skolem functions in the head.
With these additions, we can express previously inexpressible queries with ease.

The resulting language expresses all  NLOGSPACE queries, and the syntactic extension is simple and practical. It also naturally extends to GQL; in fact we provide a \emph{concrete proposal for the standards}. We plan to put our proposal on the ISO agenda, convince the committee that the proposal allows backward compatibility, and clarify which problems it solves for the standards.
Similar to the evolution of SQL standard though the 1986/89/92 versions, a similar evolution is likely for GQL and SQL/PGQ; therefore we believe that our results can have a real impact on how knowledge graphs based on the LPG model will be queried in the future.


\input{related}

\OMIT{
\subsection{Pointers; to be removed when all is written}
Some pointers

KG market 2B now, expected 25B in 10 years.
Almost everyone uses RDF or LPG with about equal split atm.
https://www.fortunebusinessinsights.com/knowledge-graph-market-11213

Also this https://www.marketreportsworld.com/market-reports/graph-database-market-14723746 gives a 60-40 to LPGs

There are areas where LPGs dominate (finance, supply chain and logistics, cybersecurity), where one needs fast traversals and analytics; areas where RDF dominates (healthcare, linked data) where reasoning and ontologies are more important. 

Bottom line: we have a split, LPG will grow. \wim
{mature? grow sounds like market share; which implies that we're saying that RDF will shrink} Both have their standards; RDF are older and have been properly scrutinized (citations).

LPGs are newer -- 2023 and 2024, analysis in its infancy. We have identified problems and here we suggest some fixes. 

We do both theory and concrete suggestions of what GQL v3 can do.

\subsection{Issues With GQL}

\begin{itemize}
\item Its path language expresses NL-complete problems but (even the CRPQ fragment?) does not capture NL. (Can we say this less abstractly? We can mention the increasing values over edges problem, which (AFAIK) leads to an open language opportunity (sliding windows). This means that ISO takes the problem seriously. Any other way to make the problem we're pointing out more concrete?)
\item No graph to graph compositionality. Mention the SIGMOD paper.
\end{itemize}

\subsection{What We Do}

}

%% file: related.tex
\paragraph{Related Work.}
The need to create new graph elements was already recognized in very
early graph database papers \cite{graphlog,good} (that precede LPGs
by decades and used a simpler graph abstraction), presenting
mechanisms for adding nodes and edges based on matches in an existing
graph. While \cite{good} showed relational completeness for a
restricted fragment and Turing completeness for the entire language,
capturing NLOGSPACE was estalished in \cite{graphlog} using (then)
recent proof of its closure under complement~\cite{immerman-book}.

In the context of LPGs, and specifically GQL and SQL/PGQ, extensions
of languages with graph element creation and with accounting for
variable bindings have been studied
\cite{filip2024,gpc,sigmod24-views,pods25,Liat2026,dominik2026} though not in the context of language design, which is our primary goal.  In
\cite{gpc}, a type system for variables in patterns was proposed, and
\cite{pods25} offered a cleaner model upon which we base our
proposal. All of \cite{dominik2026,filip2024,Liat2026} concern graph
transformations, and similarly to relational data exchange techniques \cite{ABLM14}, they use Skolemization to produce new graph elements. A proposal for defining
graph views (though not incorporating them into the queried graph) is
given in \cite{sigmod24-views}, while \cite{dominik2026} extends the \verb+MERGE+ facility
of Cypher \cite{cypher-updates} that allows creation of new
properties but is not as general as generating arbitrary graph
elements from matches. The closest in spirit to us are transformation
approaches of \cite{filip2024,Liat2026}. The former uses rules with content constructors, on top of Graph Pattern Calculus queries~\cite{gpc} and provides an implementation in Cypher. The work shows nicely how a transformation language can build on existing GQL constructs (without touching the basics of GQL). We, however, have a very different goal, which is to clean up fundamental aspects of GQL path matching and use the resulting path matching language to construct a simple and clean compositional graph-to-graph language. 
The latter paper \cite{Liat2026} stays in the realm of SQL/PGQ and requires the creation of
six relational views to generate new graph elements, which is more
cumbersome than the approach we propose and cannot be naturally
extended to a native graph language like GQL.
Transformational approaches also exist on the RDF side, for example
\cite{triple-iswc2002}.


%% file: preliminaries.tex
We assume that $\Label$, $\Key$, $\Value$, and $\Var$ are infinite sets of \emph{labels}, \emph{property names} (or \emph{keys}), \emph{values}, and \emph{variable names}, respectively. 
Moreover, given a set $A$, we denote by $\Fin(A)$ the set of all finite subsets of $A$.

\begin{definition}
A \emph{property graph} is a tuple $G = (N, E, \src, \tgt, \lab, \prop)$ where:
\begin{enumerate}
\item $N$ is a finite set of node identifiers.
\item $E$ is a finite set of edge identifiers that is disjoint from $N$.
\item $\src : E \to N$ indicates the source node (or starting node) of an edge.
\item $\tgt : E \to N$ indicates the target node (or ending node) of an edge.
\item $\lab : N \cup E \to \Fin(\Label)$ assigns a finite set of labels to each node and edge.
\item $\prop : (N \cup E) \times \Key \to \Value$ is a partial function such that if $\prop(o,k) = v$, then $v$ is the value of property $k$ for node or edge $o$. We assume that the support of $\prop$ is finite; that is, the set $\{ (o,k) \mid o \in N \cup E$, $k \in \Key \text{, and } \prop(o, k) \text{ is defined} \}$ is finite.
\end{enumerate}
\end{definition}

A \emph{graph element} in $G$ is either a node or an edge in $G$. A \emph{path} is an alternating sequence $g_1 \cdots g_n$ of nodes and edges such that, if $g_i \in E$ and $i \in  \{1,\ldots,n-1\}$, then $g_{i+1} = \tgt(g_i)$ and, moreover, if $g_i \in E$ and $i \in \{2,\ldots,n\}$, then $g_{i-1} = \src(g_i)$. In particular, paths can begin or end with an edge.  We define $\first(p) = g_1$ and $\last(p) = g_n$. A nonempty path $p$ is \emph{from $u$ to $v$} if $\first(p) = u$ and $\last(p) = v$. The empty path $\varepsilon$ is \emph{from $u$ to $u$}, for every node or edge $u$.

An \emph{annotated graph element} is a pair $(u, S)$ or a pair $[v,S]$ with $u \in N$, $v \in E$, and $S$ a finite subset of $\Var$. We call $(u, S)$ with $u \in N$ an \emph{annotated node element}, while we call $[v, S]$ with $v \in E$ an \emph{annotated edge element}. For the sake of readability, we will denote $(u, \emptyset)$ and $[v, \emptyset]$ by $(u)$ and $[v]$, respectively. For an annotated graph element $g$, we write $\el(g)$ for the graph element in $g$, that is, $\el((u,S)) = u$ and $\el([v,S]) = v$.

An \emph{annotated path} in $G$ is an alternating sequence $g_1 \cdots g_n$ of node and edge elements such that $\el(g_1) \cdots \el(g_n)$ is a path in $G$. If $n = 0$, the path is empty and we denote it by $\varepsilon$. 
Let $p$ be an annotated path in $G$. Then, similarly to paths, we define $\first(p) = o$ if either $(o, S)$ or $[o, S]$ is the first element of $p$. Similarly, $\last(p) = o$ if either $(o, S)$ or $[o, S]$ is the last element of $p$. 

\emph{Annotated path concatenation} will be fundamental to the semantics of \emph{regular path queries with variables (RPQVs)} in Section~\ref{sec:rpqvs}.
Given two graph elements $g_1$ and $g_2$, the \emph{join of $g_1$ with $g_2$}, denoted by $g_1 \bowtie g_2$, is defined as follows: 
\begin{itemize}
\item 
If $g_1 = (u, S_1)$ and
$g_2 = (u, S_2)$, then 
$g_1 \bowtie g_2 = (u, S_1 \cup S_2)$.
\item 
If $g_1 = [v, S_1]$ and
$g_2 = [v, S_2]$, then
$g_1 \bowtie g_2 = [v, S_1 \cup S_2]$.
\item If $g_1 = (u,  S_1)$ and
$g_2 = [v, S_2]$ with $\src(v) = u$, then
$g_1 \bowtie g_2 = (u,  S_1)[v, S_2]$.
\item If $g_1 = [v, S_1]$ and
$g_2 = (u,  S_2)$ with $\tgt(v) = u$, then
$g_2 \bowtie g_2 = [v, S_1](u,  S_2)$.
\end{itemize}
Notice that the join of two graph elements always produces an annotated path that respects the connectedness of elements in $G$.
In the case where $g_1 \bowtie g_2$ is defined, we say that the annotated graph elements $g_1$ and $g_2$ are \emph{joinable (in $G$)}. Given two annotated paths $p_1$ and $p_2$, the concatenation of $p_1$ with $p_2$, denoted by $\concat(p_1, p_2)$, is an annotated path defined as follows:
\begin{itemize}
\item If $p_1 = \varepsilon$, then $\concat(p_1, p_2) = p_2$.
\item If $p_2 = \varepsilon$, then $\concat(p_1, p_2) = p_1$.
\item If $p_1 = g_1 \cdots g_n$ and $p_2 = g_1' \cdots g_m'$, with $n \geq 1$ and $m \geq 1$, and $g_n$ and $g_1'$ are joinable, then 
$\concat(p_1, p_2) = g_1 \cdots g_{n-1} (g_n \bowtie g_1') g_2' \cdots g_m'$.
\end{itemize}

\begin{remark}
There is an important difference between our definition of joining (annotated) paths and the one in GQL and SQL/PGQ. Whereas we treat nodes and edges symmetrically, the standards do not. In the standards, concatenating node elements $(u,\{x\})$ and $(v,\{y\})$ will require $u = v$ and result in $(u,\{x,y\})$, which is also what we do. Concatenating edge elements $[u,\{x\}]$ and $[v,\{y\}]$ however will \emph{not} require $u = v$ and results in a path of the form $[u,\{x\}](\_,\{\})[v,\{y\}]$. That is, their semantics inserts a new node between $u$ and $v$ (which we denoted as $\_$). It is known that GQL and SQL/PGQ cannot express paths with increasing values on edges using positive combinations of \emph{path pattern expressions}~\cite{vldb25,pods25} which is their mechanism for matching paths. We believe that the reason for this inexpressibility is \emph{precisely this asymmetry between nodes and edges}. In Example~\ref{ex:increasingvalues}, we show how our approach, which we believe to be cleaner, can do it with a single RPQV, which is our counterpart of path pattern expressions.
\end{remark}

%% file: rpqv.tex
In this section, we introduce the notion of regular path query with variables (RPQV), which play a prominent role in the query language for property graphs introduced in this paper. To this end, we first need to define the conditions that can be included in such expressions. More precisely, a condition $\cond$ is defined by the following grammar, where $x, y \in \Var$, $c \in \Value$, and $k,k_1,k_2 \in \Key$:
\begin{multline*}
\cond  \quad ::= \quad x.k_1 = y.k_2  \ \mid \  x.k = c \ \mid \ x.k_1 < y.k_2 \ \mid \ x.k < c \ \mid \\
 \exists\regex  \ \mid \ (\cond \wedge \cond) \ \mid \ (\neg \cond),
\end{multline*}
Then, an RPQV expression $\expr$ is defined by the following grammar, where $\ell \in \Label$ and $x,y \in \Var$:
\begin{align*}
\node & \ ::= \	(\_) \ \mid \  (\ell)  \ \mid \  (x)  \ \mid \ (x:\ell) \\
\edge & \ ::= \	[\_]  \,\ \mid \  [\ell] \;\ \mid \  [x]  \;\ \mid \  [x:\ell]\\
\regex & \ ::= \	\varepsilon  \ \mid \  \langle\cond\rangle  \ \mid \  \node  \ \mid \ \edge  \ \mid \ 
(\regex + \regex)  \ \mid \  (\regex / \regex)  \ \mid \  (\regex^*)\\
\expr & \ ::= \regex(x,y)
\end{align*}
We refer to $\regex$ as a \emph{regular expression with variables}.

\begin{remark}
  In expressions of the form $\regex(x,y)$, the variables $x$ and $y$ will have a different role than those inside $\regex$. Indeed, whereas $x$ and $y$ will bind to nodes and edges of $G$, the variables inside $\regex$ are \emph{list variables}, i.e., they will bind to \emph{lists} of nodes and edges. List variables reflect GQL and SQL/PGQ's \emph{group variables} \cite{sigmod22,gpc,pods25}. They are also sufficiently powerful to simulate GQL and SQL/PGQ's \emph{path variables}, as we will show in Example~\ref{ex:increasingvalues}.
\end{remark}


\subsection{Semantics}
Let $G = (N, E, \src, \tgt, \lab, \prop)$ be a property graph and $p$ be an annotated path in $G$. For a variable $x$, we define $\last(p, x) = o$ if either $(o, S)$ or $[o, S]$ is the last annotated graph element of $p$ such that $x \in S$. The semantics of RPQVs is mutually recursive between the semantics of conditions $\cond$ and expressions $\regex$. We write $p_1 \xrightarrow{G,\regex} p_2$ to indicate that annotated path $p_2$ is reachable from annotated path $p_1$ in $G$ through the expression $\regex$, and write $p \models_G \cond$ to say that $p$ satisfies a condition $\cond$ in $G$.  We first define $p \models_G \cond$: 
\begin{itemize}
\item $\cond$ is $x.k_1 = y.k_2$, 
with 
$\last(p, x) = u$, $\last(p, y) = v$, $\prop(u, k_1)$ is defined, $\prop(v, k_2)$ is defined, and $\prop(u, k_1) = \prop(v, k_2)$;
\item $\cond$ is $x.k = c$ with 
$\last(p, x) = u$, $\prop(u, k)$ is defined, and $\prop(u, k) = c$;
\item $\cond$ is $x.k_1 < y.k_2$ with 
$\last(p, x) = u$, $\last(p, y) = v$, $\prop(u, k_1)$ is defined, $\prop(v, k_2)$ is defined, and $\prop(u, k) < \prop(v, \ell)$;
\item $\cond$ is $x.k < c$ with 
$\last(p, x) = u$, $\prop(u, k)$ is defined, and $\prop(u, k) < c$;
\item $\cond$ is $\exists\expr$ and there exists an annotated path $p'$ such that either $p = \varepsilon$ and $\varepsilon \xrightarrow{G, \regex} p'$, or $p \neq \varepsilon$ and $\last(p) \xrightarrow{G, \regex} p'$;
 
\item $\cond$ is $(\cond_1 \wedge \cond_2)$ and $p \models_G \cond_1$ and $p \models_G \cond_2$; or
\item $\cond$ is $(\lnot \cond_1)$ and it is not the case that $p \models_G \cond_1$.
\end{itemize}
We define the relation $p_1 \xrightarrow{G,\regex} p_2$ for in Figure \ref{fig-reach-rpqv}. Having this terminology, the set of annotated paths that matches an expression $\regex$ is defined as \[\Paths(G,\regex) := \{ p \mid \varepsilon \xrightarrow{G, \regex}  p \}\;.\] 
Finally, the semantics of $\expr = \regex(x,y)$ on $G$ is defined as follows. For $x \in \Var$ and an annotated path $p$, let $\List(x,p)$ be the list of nodes and edges in $p$ that are marked with variable $x$. Formally, we have that:
\begin{itemize}
\item $\List(x, \varepsilon) = [\,]$;

\item $\List(x, (n,S)) = [n]$ if $x \in S$ and $\List(x, (n,S)) = [\,]$ otherwise;

\item $\List(x, [e,S]) = [e]$ if $x \in S$ and $\List(x, [e,S]) = [\,]$ otherwise;

\item $\List(x, g_1 g_2 \ldots g_n) = \List(x, g_1) \cdot \List(x, g_2) \cdot \ldots \cdot L(x, g_n)$, where $g_1 g_2 \ldots g_n$ is an annotated path and $\cdot$ is the concatenation operator for lists.
\end{itemize}
For example, we have that:
\begin{align*}
\List(x, (n_1, \{x,y\})[e_1,\{z\}](n_1,\{x\})[e_2,\{x,z\}](n_2,\{x,y,z\})) \ = \ [n_1,n_1,e_2,n_2].
\end{align*}
Then, an RPQV returns a \emph{set of bindings} $f$ for the variables in the expression:
\begin{multline*}
    \sem{\regex(x,y)}_G \ := \ \{f \mid \dom(f) = \dom(\regex) \cup \{x,y\}, \text{ and }\\ 
    \exists p \in \Paths(G,\regex) 
    \text{ from $f(x)$ to $f(y)$ such that}\\ 
    f(u) = \List(p,u) \text{ for all } u \in \dom(\regex)\}.
\end{multline*}

\begin{figure}[t]
\begin{center}
\begin{tabular}{lcp{85mm}}
$p_1 \xrightarrow{G, \varepsilon}  p_2$ &\ \ \ \ \ & $\text{if } p_1 = p_2$\\
$p_1 \xrightarrow{G, \langle\cond\rangle} p_2$  && $\text{if } p_1 \models_G \cond \text{ and } p_1 = p_2$\\
$p_1 \xrightarrow{G, (\_)}  p_2$ && $\text{if } p_2 = \concat(p_1, (u)) \text{, where } u \in N$\\
$p_1 \xrightarrow{G, (\ell)}  p_2$ && $\text{if } p_2 = \concat(p_1, (u)) \text{, where } u \in N \text{ and } \ell \in \lab(u)$\\
$p_1 \xrightarrow{G, (x)}  p_2$ && $\text{if } p_2 = \concat(p_1, (u, \{x\})) \text{, where } u \in N$\\
$p_1 \xrightarrow{G, (x:\ell)}  p_2$ && $\text{if } p_2 = \concat(p_1, (u, \{x\})) \text{, where } u \in N \text{ and}$ \mbox{$\ell \in \lab(u)$}\\
$p_1 \xrightarrow{G, [\_]}  p_2$ && $\text{if } p_2 = \concat(p_1, [v]) \text{, where } v \in E$\\
$p_1 \xrightarrow{G, [\ell]}  p_2$ && $\text{if } p_2 = \concat(p_1, [v]) \text{, where } v \in E \text{ and } \ell \in \lab(v)$\\
$p_1 \xrightarrow{G, [x]}  p_2$ && $\text{if } p_2 = \concat(p_1, [v, \{x\}]) \text{, where } v \in E$\\
$p_1 \xrightarrow{G, [x:\ell]}  p_2$ && $\text{if } p_2 = \concat(p_1, [v, \{x\}]) \text{, where } v \in E \text{ and}$ \mbox{$\ell \in \lab(v)$}\\
$p_1 \xrightarrow{G, \regex_1 + \regex_2}  p_2$ && $\text{if } p_1 \xrightarrow{G, \regex_1}  p_2 \text{ or } p_1 \xrightarrow{G, \regex_2}  p_2$\\
$p_1 \xrightarrow{G, \regex_1/\regex_2}  p_2$  && $\text{if there exists a path $p_3$ such that } p_1 \xrightarrow{G, \regex_1}  p_3$\\
&&$\text{and } p_3 \xrightarrow{G, \regex_2}  p_2$\\
$p_1 \xrightarrow{G, \regex^*}  p_2$	&& $\text{if there exists } n \in \mathbb{N} \text{ such that } p_1 \xrightarrow{G, \regex^n}  p_2$, where\\ &&$\regex^0 = \varepsilon$ and $\regex^{m} = \regex/\regex^{m-1}$ for every \hbox{$m \geq 1$}.
\end{tabular}
\end{center}
\caption{Definition of reachability among annotated paths for RPQVs.\label{fig-reach-rpqv}}
\end{figure}

\subsection{Notation and Examples}\label{sec:notation} In an expression of the form $\regex(x,y)$, we refer to $x$ and $y$ as \emph{boundary variables}. We  will use GQL-style infix notation and write $\regex(x,y)$ as $(x)\ \regex\ (y)$, or $[x]\ \regex\ (y)$, etc., depending on whether the boundary variables are nodes or edges. 
This notation cannot express all RPQVs because, in general, it can happen that $x$ binds to a node for some answers and to an edge for other answers, but it is sufficient for all our examples.
Furthermore, we omit the explicit concatenation operator / to improve readability. This is standard in formal languages and is also done in GQL and SQL/PGQ.

\begin{example}\label{ex:increasingvalues}
    The RPQV 
    $$(x) \ ((z_1)[\_](z_2)\langle z_1.\mathit{value} < z_2.\mathit{value} \rangle )^* \ (y)$$ matches annotated paths from node $x$ to node $y$ with \emph{increasing values on nodes}. To understand this, notice that the variable $z_2$ of one iteration joins with the variable $z_1$ from the next iteration (just like in Cypher and GQL). 
    More formally, this expression returns bindings $f$ such that there exists an annotated path $p$ from $f(x)$ to $f(y)$ such that $\mathit{value}$ increases between every pair of consecutive nodes on $p$. Moreover, $f(z_1)$ is the list of all nodes on $p$ except for the last, and $f(z_2)$ is the list of all nodes on $p$ except for the first. If we want to have a single variable that contains all the nodes in $p$, we could write $(x) ((z_1)[\_](z_2)\langle z_1.\mathit{value} < z_2.\mathit{value} \rangle )^* (z_1) (y)$, so that the list of all nodes in the annotated path is stored in $f(z_1)$. Notice that the extra expression $(z_1)$ forces the last node in the annotated path to be concatenated to the list for variable $z_1$ computed by the expression $((z_1)[\_](z_2)\langle z_1.\mathit{value} < z_2.\mathit{value} \rangle )^*$, which contains all nodes except the last one.
    
    For \emph{increasing values on edges} we can use the RPQV
    $$[x]\ ([z_1](\_)[z_2] \langle z_1.\mathit{value} < z_2.\mathit{value} \rangle )^*\ [y]\;.$$
    Notice that this RPQV is not expressible as a GQL or SQL/PGQ path pattern expression~\cite{vldb25}. For us, however, it is completely dual to the RPQV for the increasing values on nodes condition, which is desirable, since nodes and edges should have the same status in the language. If one wants the latter expression to return a node-to-node path, one can write
    $$(x)\ (\_)\;([z_1](\_)[z_2] \langle z_1.\mathit{value} < z_2.\mathit{value} \rangle )^*\;(\_)\ (y)\;.$$
\end{example}

\begin{remark}
    RPQVs can match entire paths in a single variable, which means that they can fulfill the role of GQL's path variables. To do this, we can simply use a fresh variable that we use to annotate every element. For instance, the expression 
    $([z_1][z](\_)(z)[z_2][z] \langle z_1.\mathit{value} < z_2.\mathit{value} \rangle )^*$, obtained from the increasing values on edges RPQV by adding $[z]$ or $(z)$ right after each (possibly anonymous) variable inside its $\regex$, binds the entire paths to the variable $z$. 
\end{remark}



\subsection{Number of Output Paths and Complexity}
We note that $\sem{\regex(x,y)}_G$ can be infinite if the graph $G$ has cycles.
Practical languages solve this issue by imposing that the paths $p$ should be \emph{shortest}, \emph{simple} (no repeating nodes) or \emph{trails} (no repeating edges)~\cite{sigmod22,gpc,icdt23,cypher}. The same restrictions can be applied to RPQVs and are independent of the design of RPQVs themselves. In fact, it makes much sense to study RPQVs \emph{without} these restrictions, because evaluation problems for RPQs without list variables are typically in NLOGSPACE~\cite{CruzMW87,MendelzonW95}. These problems become NP-complete if constraints such as simple paths~\cite{MendelzonW95,BaganBG20} and trails~\cite{MartensNP23} are added, even on undirected graphs~\cite{MartensP22}. 

\sloppy
Ideally, we would therefore have RPQVs that can be evaluated in NLOGSPACE, even if the mechanism for forcing them to match a finite number of results may render evaluation NP-complete. An additional argument for our approach is that it is well-known that compact representations of the infinitely many paths that match RPQs can be computed in linear time~\cite{pmr,pathfinder}. This means that a query engine could internally use this representation (similar to how we implement factorized databases~\cite{OlteanuS16}) and we do not need to force their result set to be finite.

Regarding complexity, for each fixed RPQV $\expr$, consider the following problem $\mathsf{Eval}(\expr)$: Given a graph $G$ and binding $f$, is $f \in \sem{\expr}$?

\begin{proposition}\label{prop:1}
$\mathsf{Eval}(\expr)$ is in NLOGSPACE for each RPQV $\expr$.
\end{proposition}

\begin{proof}
\emph{Sketch}. The idea is to encode both $G$ and $f$ so that the evaluation problem becomes expressible in FO(TC), first-order with transitive closure, which is known to be evaluable in NL \cite{immerman-book}. Graphs are already relational structures over the universe that contains $N\cup E$ and the values present as property values. We further extend the universe with a disjoint ordered set $p_1,\ldots,p_n$ of positions in lists, with $n$ being the maximum length of a list in the range of $f$. Each such list will be encoded as a set of pairs $(p_1,e_1),(p_2,e_2)$, etc, indicating a position of a graph element in the list, with $p_1 \prec p_2 \prec \cdots \prec p_n$. With this, and the access to $\prec$ in addition to predicates defining $G$ and $f$, it is routine to encode the evaluation problem in FO(TC). \hfill $\Box$
\end{proof}

%% file: hash-datalog.tex
$\hashd$ (pronounced \emph{hash-Datalog}) is a simple graph transformation language that uses Datalog with safe negation and with RPQVs in the bodies. A \emph{$\hashd$ program} is a sequence of \emph{computation} and \emph{update} programs. The role of computation programs is to compute the necessary information for defining a new graph, including the IDs of new nodes and edges to be generated. Update programs specify the concrete relations $\node$, $\edge$, $\src$, $\tgt$, $\lab$, and $\prop$ that constitute the new graph. 
In order to create new node IDs and edge IDs, rules
\begin{align}\label{rule-dat}
A(\bar x) \ \gets \ B_1(\bar y_1), \ldots, B_n(\bar y_n), \lnot C_1(\bar z_1), \ldots, \lnot C_m(\bar z_m) \tag{\dag}
\end{align}
in computation programs will give us access to two relations: the relation $A$, which is obtained using standard Datalog semantics and the relation $\#A$ which, for each tuple in $A$, stores an identifier (using Skolemization).

\begin{remark}\label{remark:hash-tables}
At first sight, it may seem expensive to produce both relations $A$ and $\#A$. But this is actually not the case in systems that implement GQL or SQL/PGQ. Indeed, such systems typically produce a table $T_A$ for the tuples in $A$, and the extra value needed for $\#A$ can easily be obtained by taking, e.g., the internal ID for the respective tuple in $T_A$ which every DBMS will provide. 
\end{remark}

\subsection{A Guided Tour of \hashd}
We now look at a few examples that illustrate the capabilities of $\hashd$.
For space reasons, we provide its fully formal definition in Appendix~\ref{app:hashd}.

\paragraph{A Simple Graph Transformation.}
Assume that $G_1$ is the property graph

\centerline{\resizebox{.6\linewidth}{!}{%
    \begin{tikzpicture}
        \tikzstyle{every state}=[draw,thick,rectangle,rounded corners,fill=white!85!black,minimum size=5mm, text=black, font=\ttfamily, inner sep=0pt]
        \tikzstyle{every node}=[font=\ttfamily]

        \node[state] (n1) at (-4,-0) {};
        \node[state] (n2) at (-2,0) {};
        \node[state] (n3) at (2,0) {};
        \node[state] (n4) at (4,-0) {};

        \node[state,dashed,fill=white!95!black,inner sep=2pt] (e1) at (0,0) {
            \begin{tabular}{l}
                \small {\bf k:} 1
            \end{tabular}
        };

        \node[rectangle,draw,dashed,fill=white] (e1label) at (-3,0) {\ttfamily \small $a$};
        \node[rectangle,draw,dashed,fill=white] (e2label) at ($(e1)+(0.4,0.4)$) {\ttfamily \small $b$};
        \node[rectangle,draw,dashed,fill=white] (e3label) at (3,0) {\ttfamily \small $a$};
      
    \begin{pgfonlayer}{background}        
        \path[-latex,thick]
        (n1) edge (n2)
        (n2) edge (n3)
        (n4) edge (n3)
        ;
          
    \end{pgfonlayer}

    \end{tikzpicture}
    }
}

\noindent with four nodes and three edges. Two edges are labeled $a$ and one has property $k$ with value 1 and is labeled $b$. From $G_1$, we want to obtain a new graph $G_2$ where every $a$-labeled edge in $G_1$ becomes a new node with label $c$, and there is an edge with label $d$ from a node $u_1$ to node $u_2$ if in $G_1$ there was an edge with label $b$ from the end-node of the edge represented by $u_1$ to the end-node of the edge represented by $u_2$. In a picture (that still contains $G_1$ in black), the new graph $G_2$ (in blue) should be obtained as follows:

\centerline{\resizebox{.6\linewidth}{!}{%
    \begin{tikzpicture}
        \tikzstyle{every state}=[draw,thick,rectangle,rounded corners,fill=white!85!black,minimum size=5mm, text=black, font=\ttfamily, inner sep=0pt]
        \tikzstyle{every node}=[font=\ttfamily]

        \node[state] (n1) at (-4,-0) {};
        \node[state] (n2) at (-2,0) {};
        \node[state] (n3) at (2,0) {};
        \node[state] (n4) at (4,-0) {};

        \node[state,dashed,fill=white!95!black,inner sep=2pt] (e1) at (0,0) {
            \begin{tabular}{l}
                \small {\bf k:} 1
            \end{tabular}
        };

        \node[state,dashed,draw=blue!80,fill=white!95!black,inner sep=2pt] (ne1) at (0,1.2) {
            \begin{tabular}{l}
                \small {\bf k:} 1
            \end{tabular}
        };

        \node[rectangle,draw,dashed,fill=white] (e1label) at (-3,0) {\ttfamily \small $a$};
        \node[rectangle,draw,dashed,fill=white] (e2label) at ($(e1)+(0.4,0.4)$) {\ttfamily \small $b$};
        \node[rectangle,draw,dashed,fill=white] (e3label) at (3,0) {\ttfamily \small $a$};

        \node[rectangle,draw,dashed,fill=white] (ne1label) at ($(ne1)+(0.4,0.4)$) {\ttfamily \small $d$};

    \begin{pgfonlayer}{background}        

        \node[draw=blue!80,thick,fill=blue!20,inner sep=8pt,rectangle,rounded corners,fit=(n1) (n2)] (nn1) {};
        \node[draw=blue!80,thick,fill=blue!20,inner sep=8pt,rectangle,rounded corners,fit=(n3) (n4)] (nn2) {};
        
        \path[-latex,thick]
        (n1) edge (n2)
        (n2) edge (n3)
        (n4) edge (n3)
        ;

        \path[-latex,thick,draw=blue!80]
        (nn1) edge [bend left] (nn2);

    \end{pgfonlayer}

    \node[rectangle,draw=blue!80,fill=white] (nn1label) at ($(nn1)+(0.5,0.65)$) {\ttfamily \small $c$};
    \node[rectangle,draw=blue!80,fill=white] (nn2label) at ($(nn2)+(0.7,0.65)$) {\ttfamily \small $c$};

    \end{tikzpicture}
    }}
    
\noindent Furthermore, the edges with label $d$ in $G_2$ should inherit the values of property $k$ from the original edges in $G_1$. So, the new graph $G_2$ will be

\centerline{\resizebox{.3\linewidth}{!}{%
    \begin{tikzpicture}
        \tikzstyle{every state}=[draw,thick,rectangle,rounded corners,fill=white!85!black,minimum size=5mm, text=black, font=\ttfamily, inner sep=0pt]
        \tikzstyle{every node}=[font=\ttfamily]

        \node[state] (n1) at (-2,0) {};
        \node[state] (n2) at (2,0) {};

        \node[rectangle,draw,fill=white] (nn1label) at ($(n1)+(0.3,0.3)$) {\ttfamily \small $c$};
        \node[rectangle,draw,fill=white] (nn2label) at ($(n2)+(0.3,0.3)$) {\ttfamily \small $c$};

        \node[state,dashed,fill=white!95!black,inner sep=2pt] (e1) at (0,0) {
            \begin{tabular}{l}
                \small {\bf k:} 1
            \end{tabular}
        };

        \node[rectangle,draw,dashed,fill=white] (e1label) at ($(e1)+(0.4,0.4)$) {\ttfamily \small $d$};
      
    \begin{pgfonlayer}{background}        
        \path[-latex,thick]
        (n1) edge (n2)
        ;
          
    \end{pgfonlayer}

    \end{tikzpicture}
    }
}

\noindent
A $\hashd$ program doing this would first define predicates for generating identifiers for the nodes and edges in $G_2$ in its \emph{computation program}:

\smallskip
$
\begin{array}{rl}
A(x, e, y) & \gets (x)\ (\_)[e\!:\!a](\_)\ (y)\\
B(x, y, z) & \gets \#A(u_1, e_1, v_1, x), \#A(u_2, e_2, v_2, y),\  (v_1)\;(\_)[z\!:\!b](\_)\;(v_2)
\end{array}
$
\smallskip

\noindent Intuitively, $A$ has triples $(x,e,y)$ such that $e$ is an $a$-labeled edge from $x$ to $y$ in $G_1$. (Different edges $e$ in $G_1$ will yield different tuples in $A$.) Furthermore, $\#A$ has quadruples $(x,e,y,h)$ such that $h$ is unique for each such combination of $(x,e,y)$. That is, we can think of the values $h$ as new node IDs in $G_2$. Each $a$-labeled edge from $x$ to $y$ in $G_1$ will generate a different node in $G_2$, because we included the variable $e$ in the definition of $A$. If we would have omitted $e$, then we would generate at most one node in $G_2$ for each node pair $(x,y)$ in $G_1$.
Moreover, we use the new node IDs to define $B$, which intuitively has triples $(x,y,z)$ such that $x$ and $y$ are $a$-labeled edges in $G_1$ -- expressed as $\#A(u_1, e_1, v_1, x)$ and $\#A(u_2, e_2, v_2, y)$ -- and $z$ is a $b$-labeled edge from $v_1$ to $v_2$.

We now define the property graph predicates in the \emph{update program}:

\smallskip
$
\begin{array}{rl@{\hspace{2cm}}rl}
\node(h) & \gets \#A(x, e, y, h) & \lab(x, y) & \gets \node(x), y = ``c" \\
\edge(h) & \gets \#B(u, v, w, h) & \lab(x, y) & \gets \edge(x), y = ``d"\\
\src(e, u) & \multicolumn{3}{l}{\gets \edge(e), \node(u), \#B(u, z_2, z_3, e)}\\
\tgt(e, v) & \multicolumn{3}{l}{\gets \edge(e), \node(v), \#B(z_1, v, z_3, e)}\\
\prop(x, y, z) & \multicolumn{3}{l}{\gets  \edge(x), y = ``k", \#B(y_1, y_2, y_3, x), z = y_3.k}
\end{array}
$


\paragraph{The Dual Graph (and Increasing Values on Edges Revisited).} 
$\hashd$ can copy labels or key/value pairs from the input, for example to compute the 
\emph{dual graph} $G^*$ of any given property graph $G$ (which is obtained from $G$ by turning nodes into edges and vice versa). In $\hashd$ we can do it as follows:


\smallskip
$
\mathcal{C}: \qquad
\begin{array}{rl@{\hspace{1cm}}rl}
N(e) & \gets [e] & 
L(x,y) & \gets \lab(x,y)\\
E(e_1,x,e_2) & \gets [e_1]\;[\_](x)[\_]\;[e_2] &
P(x,y,z) & \gets \prop(x,y,z)
\end{array}
$

\smallskip
$
\mathcal{U}: \qquad
\begin{array}{rl@{\hspace{2cm}}rl}
\node(h) & \gets \#N(e, h) & \lab(x, y) & \gets L(x,y)\\
\edge(h) & \gets \#E(e_1, x, e_2, h) & \prop(x, y, z) & \gets  P(x,y,z)\\
\src(e, u) & \multicolumn{3}{l}{\gets \edge(e), \node(h_u), \#N(u,h_u), \#E(u, z_2, z_3, e)}\\
\tgt(e, v) & \multicolumn{3}{l}{\gets \edge(e), \node(h_v), \#N(v,h_v), \#E(z_1, z_2, v, e)}
\end{array}
$

\smallskip
\noindent
The computation program $\mathcal{C}$ defines hash predicates $\#N$ and $\#E$ for generating the node and edge IDs in the dual graph $G^*$, and copies the label and property information in $L$ and $P$. The update program $\mathcal{U}$ uses these to define all the components of the new graph $G^*$. Notice that $\prop$ in the computation program refers to the properties of $G$, whereas in the update program it is used to populate $G^*$. This is why we keep the computation and update programs separate.

Interestingly, the RPQV in Example
\ref{ex:increasingvalues}
can be evaluated on the dual graph $G^*$ to determine whether there is a path from $x$ to $y$ with increasing node values, which corresponds to a path with increasing edge values in $G$.



\paragraph{Turning Paths into Edges.}
Assume we are given a graph $G_1$ and we want to construct a new graph $G_2$ that consists only of the nodes with owners ``Mike'' and ``Megan'' in $G_1$. Furthermore, each path from ``Mike'' to ``Megan'' in $G_1$ such that $\mathit{value}$ increases along edges should become an edge in $G_2$.

\smallskip
$
\mathcal{C}: \qquad
\begin{array}{rl}
\mathit{Megan}(x) & \gets (x)\; (y)\langle y.\mathit{owner} = ``\text{Megan}"\rangle \; (x)\\
\mathit{Mike}(x) & \gets (x)\;(y)\langle y.\mathit{owner} = ``\text{Mike}"\rangle\; (x)\\
\mathit{Path}(x, y, z) & \gets 
(x)\ (\_)([z](\_)[z'] \langle z.\mathit{value} < z'.\mathit{value} \rangle )^*[z](\_)\  (y)\\[2mm]
\end{array}
$

$
\mathcal{U}: \qquad
\begin{array}{rl@{\hspace{1cm}}rl}
\node(x) & \gets  \#\mathit{Megan}(y, x) & 
\edge(p) & \gets  \#\mathit{Path}(x, y, z, p)\\
\node(x) & \gets  \#\mathit{Mike}(y, x) & &\\
\src(p, y) & \multicolumn{3}{l}{\gets  \edge(p), \node(y), \#\mathit{Path}(z_1, z_2, z_3, p), \#\mathit{Megan}(z_1, y)}\\
\tgt(p, y) & \multicolumn{3}{l}{\gets  \edge(p), \node(y), \#\mathit{Path}(z_1, z_2, z_3, p), \#\mathit{Mike}(z_3, y)}\\
\end{array}
$

\smallskip
\noindent 
Examples such as this one become significantly more interesting when we add aggregation over paths to the language. GQL and SQL/PGQ allow  this, and it is easy to extend our proposal with it (see Appendix~\ref{app:aggregation}). Essentially, the extension will allow us to write 
$\prop(x, y, z) \gets \edge(x), y = ``\text{length}", \#B(z_1, z_2, z_3, x), z = \text{length}(z_2)$ if we want to add a property ``length'' to each edge in the output of the previous program that has the length of the corresponding path.




\paragraph{Increasing Values on Both Nodes and Edges.} 
Whereas GQL and SQL/PGQ path pattern expressions can match paths with increasing values on nodes, but cannot match paths with increasing values on edges, $\hashd$ can \emph{even match paths that have both properties}. This is not surprising once we know that $\hashd$ can express all properties in NLOGSPACE, but it is instructive to see how it can be done.
In fact, $\hashd$ can do it in two very different ways, both of which we will show here.

The first way is with a single RPQV. We show it with an RPQV that starts in a node and ends in an edge.
\[(x) \ (u_1)[v_1]\big((u_2) \langle u_1.\mathit{value} < u_2.\mathit{value} \rangle (u_1) [v_2] \langle v_1.\mathit{value} < v_2.\mathit{value}\rangle [v_1] \big)^*\ [y] \]
Intuitively, this RPQV works as follows. Recall that $x$ and $y$ simply bind to the first node and last edge in the matched path, respectively. Variables $v_1$, $v_2$, $u_1$, $u_2$ will bind to lists. The expression matches the first node we see in variable $u_1$ and the first edge in $v_1$. Then, an iteration starts in which we repeatedly do the following:
\begin{itemize}
\item add the next node to the list for $u_2$;
\item check if the last node in the list for $u_1$ has a smaller value than the node we just added to the list for $u_2$;
\item add the current node to the list for $u_1$; 
\item add the next edge to the list for $v_2$;
\item check if the last edge in the list for $v_1$ has a smaller value than the edge we just added to the list for $v_2$;
\item add the current edge to the list for $v_1$; 
\end{itemize}
Notice how our definition of annotated path concatenation is crucial for how this RPQV works. Using the node and edge collapsing mechanism (called \emph{joining} in Section~\ref{sec-prelim}), we stay in the same graph element until the RPQV switches to a different kind of element (i.e., a switch from node to edge or vice versa).

The second way relies much more significantly on the power of combining Datalog and element creation. In the computation program, we can hash node/edge pairs and edge/node pairs so that we can turn them into new nodes in the new graph. For readability, we use $\_$ do denote variables that do not join.

\smallskip
$
\begin{array}{rl}
\mathit{NE}(n,e) &\gets \src(e,n)\\
\mathit{EN}(e,n) &\gets \tgt(e,n)\\
\mathit{GoodEdge}(ne,en) & \gets \#\mathit{NE}(n_1,e,ne), \#\mathit{EN}(e,n_2,en), n_1.\mathit{value} < n_2.\mathit{value}\\
\mathit{GoodEdge}(en,ne) & \gets \#\mathit{EN}(e_1,n,en), \#\mathit{NE}(n,e_2,ne),  e_1.\mathit{value} < e_2.\mathit{value}\\
\end{array}
$


\medskip

$
\begin{array}{rl}
\node(h) &  \gets \#\mathit{NE}(\_,\_,h)\\
\node(h)  & \gets \#\mathit{EN}(\_,\_,h)\\
\edge(e) & \gets \#\mathit{GoodEdge}(\_,\_,e)\\
\src(e,s) & \gets \#\mathit{GoodEdge}(s,\_,e)\\
\tgt(e,t) & \gets \#\mathit{GoodEdge}(\_,t,e)\\
\end{array}$

\medskip
On the resulting graph, the answers that start in a node and end with an edge can now be obtained by the rule

\smallskip 
$
\begin{array}{rl}
\mathit{Answer}(n,e)  \gets & (id_1)\; ((\_) [\_] (\_))^*\;  (id_2),\\
                            & \#\mathit{NE}(n,\_,id_1), \#\mathit{NE}(\_,e,id_2)
\end{array}
$

having access to the hash predicates. The other combinations for start and end of paths (node/node, edge/edge, edge/node) are similar.

\subsection{The Formal Definition of \hashd}\label{app:hashd}






To define $\hashd$ programs, we first need to define the notions of \emph{computation} and \emph{update programs}. A computation program is a set of Datalog rules defined over a property graph that produces a set of intensional predicates. An update program is a set of rules defined over those intensional predicates that produces a property graph. In this way, a sequence of computation/update programs produces a sequence of property graphs.


\newcommand{\comp}{\Pi_{\textit{comp}}}
\newcommand{\upd}{\Pi_{\textit{upd}}}

\subsubsection{Computation Programs.} Formally, a \emph{computation program} $\comp$ is a non-recursive Datalog program whose rules are of the form \eqref{rule-dat}, where (i) each extensional atom $B_i(\bar y_i)$ and each extensional atom $C_j(\bar z_j)$ is either an RPQV expression or one of the relational atoms $\node(x)$, $\edge(x)$, $\src(x, y)$, $\tgt(x, y)$, $\lab(x, y)$, $\prop(x, y, z)$ that define the components of a property graph; (ii) $\bar x$, $\bar y_1$, $\ldots$, $\bar y_n$, $\bar z_1$, $\ldots$, $\bar z_m$ are tuples of variables such that $\bar x \subseteq \bar y_1 \cup \cdots \cup \bar y_n$ and $\bar z_1 \cup \cdots \cup \bar z_m \subseteq \bar y_1 \cup \cdots \cup \bar y_n$;\footnote{Slightly abusing notation, we also use set terminology for tuples of variables. Hence, for example, we use notation $\bar x \cup \bar y$ to define a set of variables consisting of the variables occurring in $\bar x$ or $\bar y$, and we use notation $\dom(f) = \bar x$ to indicate that the domain of $f$ is the set of variables occurring in the tuple $\bar x$.} and (iii) no list variable occurs in two of more of the sequences $\bar y_1$, $\ldots$, $\bar y_n$, $\bar z_1$, $\ldots$, $\bar z_m$. Notice that the second condition only allows rules with safe negation, while the third condition enforces joins of predicates in the body of a rule to be performed on boundary variables (cf.\ Section~\ref{sec:notation}), not on list variables. 

To define the semantics of a computation program, we start with the evaluation of relational atoms over a property graph $G = (N, E, \src, \tgt, \lab, \prop)$:
\begin{align*}
\sem{\node(x)}_G \ &= \ \{ f \mid \dom(f) = \{x\} \text{ and } f(x) \in N \}\\
\sem{\edge(x)}_G \ &= \ \{ f \mid \dom(f) = \{x\} \text{ and } f(x) \in E \}\\
\sem{\src(x, y)}_G  \ &= \{ f \mid \dom(f) = \{x,y\} \text{ and } \src(f(x)) = f(y) \}\\
\sem{\tgt(x, y)}_G  \ &= \{ f \mid \dom(f) = \{x,y\} \text{ and } \tgt(f(x)) = f(y) \}\\
\sem{\lab(x, y)}_G  \ &= \{ f \mid \dom(f) = \{x,y\} \text{ and } f(y) \in \lab(f(x)) \}\\
\sem{\prop(x, y, z)}_G  \ &= \{ f \mid \dom(f) = \{x,y,z\} \text{ and } \prop(f(x),f(y)) = f(z) \}
\end{align*}
Each rule \eqref{rule-dat} defines an intensional predicate $A$ as well as an intensional hash predicate $\#A$. Such a hash predicate $\#A$ stores an identifier for each tuple in $A$, which is used when generating a new property graph. Formally,
from now on we assume that $H$ is a fixed hash function, and use notation $H(A(\bar x))$ to indicate that $H$ is applied to the string representation of $A(\bar x)$. Then, 
given a property graph $G = (N, E, \src, \tgt, \lab, \prop)$, the evaluation of rule \eqref{rule-dat} produces the following sets of facts, assuming that $y$ is a fresh variable not occurring in rule~\eqref{rule-dat}:
\begin{align*}
&\sem{A(\bar x) \gets B_1(\bar y_1), \ldots, B_n(\bar y_n), \neg C_1(\bar z_1), \ldots, \neg C_m(\bar z_m) }_G \ = \\
&\hspace{20pt}\{ A(\bar a) \mid \exists g \text{ such that }
\dom(g) = \bar y_1 \cup \cdots \cup \bar y_n,\\ 
&\hspace{54pt}g|_{\bar y_i} \in \sem{B_i(\bar y_i)}_G \text{ for all } i \in \{1, \ldots, n\},\\
&\hspace{54pt}g|_{\bar z_j} \not\in \sem{C_j(\bar z_j)}_G \text{ for all } j \in \{1, \ldots, m\}, \text{ and } \bar a = g(\bar x)\},\\
&\sem{\#A(\bar x,y) \gets B_1(\bar y_1), \ldots, B_n(\bar y_n), \neg C_1(\bar z_1), \ldots, \neg C_m(\bar z_m)}_G \ = \\
&\hspace{20pt}\{ A(\bar a, b) \mid A(\bar a) \in \sem{A(\bar x) \gets B_1(\bar y_1), \ldots, B_n(\bar y_n), \neg C_1(\bar z_1), \ldots, \neg C_m(\bar z_m)}_G\\ 
&\hspace{63pt}\text{and } b = H(A(\bar a)) \}
\end{align*}
Finally, the evaluation of a computation program $\comp$ over a property graph $G$, denoted by $\sem{\comp}_G$, is defined as:
\begin{align*}
&\bigcup_{A(\bar x) \gets B_1(\bar y_1), \ldots, B_n(\bar y_n), \neg C_1(\bar z_1), \ldots, \neg C_m(\bar z_m) \in \comp} \\
&\hspace{50pt}\bigg(\sem{A(\bar x) \gets B_1(\bar y_1), \ldots, B_n(\bar y_n), \neg C_1(\bar z_1), \ldots, \neg C_m(\bar z_m)}_G
\ \cup\\
&\hspace{58pt}\sem{\#A(\bar x, y) \gets B_1(\bar y_1), \ldots, B_n(\bar y_n), \neg C_1(\bar z_1), \ldots, \neg C_m(\bar z_m)}_G\bigg)
\end{align*}
Notice that $\sem{\comp}_G$ is a set of facts, that is, a set of ground atoms of the form $A(\bar a)$ or $\#A(\bar a, b)$. We refer to such sets as \emph{relational instances}.


\subsubsection{Update Programs.} An \emph{update program} $\upd$ is a non-recursive Datalog program whose rules are of the form \eqref{rule-dat}, where
(i) $A(\bar x)$ is any of the relational atoms $\node(x)$, $\edge(x)$, $\src(x, y)$, $\tgt(x, y)$, $\lab(x, y)$, $\prop(x, y, z)$; (ii) each $B_i(\bar y_i)$ and each $C_j(\bar z_j)$ is either a relational atom or a relational hash atom; (iii) $\bar x$, $\bar y_1$, $\ldots$, $\bar y_n$, $\bar z_1$, $\ldots$, $\bar z_m$ are tuples of variables such that $\bar x \subseteq \bar y_1 \cup \cdots \cup \bar y_n$ and $\bar z_1 \cup \cdots \cup \bar z_m \subseteq \bar y_1 \cup \cdots \cup \bar y_n$; and (iv) no list variable occurs in two of more of the sequences $\bar y_1$, $\ldots$, $\bar y_n$, $\bar z_1$, $\ldots$, $\bar z_m$. 

Given a relational instance $I$, the evaluation of rule \eqref{rule-dat}, which is denoted by $\sem{A(\bar x) \ \gets \ B_1(\bar y_1), \ldots, B_n(\bar y_n), \lnot C_1(\bar z_1), \ldots, \lnot C_m(\bar z_m)}_I$, is defined as usual, considering that this is a standard Datalog rule evaluated over a relational instance. Then the evaluation of an update program $\upd$ over a relational instance $I$, denoted by $\sem{\upd}_I$, is defined as a property graph $G = (N, E, \src, \tgt, \lab, \prop)$ such that:
\begin{itemize}
\item $n \in N$ if and only there exists a rule $r \in \upd$ such that $\node(n) \in \sem{r}_I$.

\item $e \in E$ if and only if there exists a rule $r \in \upd$ such that $\edge(e) \in \sem{r}_I$.

\item $\src(e) = n$ if and only $n$ is the only element for which there exists a rule $r \in \upd$ such that $\src(e, n) \in \sem{r}_I$. If $e \not\in E$ or $n \not\in N$, then~$\sem{\upd}_I = \emptyset$.

\item $\tgt(e) = n$ if and only if $n$ is the only element for which there exists a rule $r \in \upd$ such that $\tgt(e, n) \in \sem{r}_I$. If $e \not\in E$ or $n \not\in N$, then~$\sem{\upd}_I = \emptyset$. Moreover, if $\src(e) = n$ but $\tgt(e)$ is not defined, or $\tgt(e) = n$ but $\src(e)$ is not defined, then $\sem{\upd}_I = \emptyset$.

\item $\ell \in \lab(o)$ if and only if there exists a rule $r \in \upd$ such that $\lab(o, \ell) \in \sem{r}_I$. If $o \not\in N \cup E$ or $\ell \not\in \Label$, then~$\sem{\upd}_I = \emptyset$.

\item $\prop(o, k) = v$ if and only if $v$ is the only element for which there exists a rule $r$ such that
$\prop(o, k, v) \in \sem{r}_I$. If $o \not\in N \cup E$, or $k \not\in \Key$, or $v \not\in \Value$, then 
$\sem{\upd}_I = \emptyset$.
\end{itemize}
With all this terminology in place, a $\hashd$ program (pronounced \emph{hash-Datalog}) is defined as a sequence $\Pi = (\comp^1, \upd^1, \ldots, \comp^k, \upd^k)$ such that $\comp^i$ is a computation program and $\upd^i$ is an update program, for every $i \in \{1, \ldots, k\}$. The evaluation of such a $\hashd$ program $\Pi$ over a property graph $G$ is defined by considering sequences $\{I_i\}_{i \in \{1, \ldots, k\}}$ and $\{G_i\}_{i \in \{1, \ldots, k\}}$ of relational instances and property graphs, respectively. More precisely, $I_1 = \sem{\comp^1}_{G}$, $G_1 = \sem{\upd^1}_{I_1}$, and for every $i \in \{2, \ldots, k\}$:
\begin{align*}
I_i \ &=\ \sem{\comp^i}_{G_{i-1}}\\
G_i \ &=\ \sem{\upd^i}_{I_i}
\end{align*}
Then we have that $\sem{\Pi}_G = G_k$.

\subsection{No expressiveness holes} We conclude this section by showing that the unusual expressiveness gaps described in the introduction do not arise in $\hashd$.

\begin{proposition}\label{prop:2}
$\hashd$ can express every query in NLOGSPACE.
\end{proposition}

\begin{proof}
   Note that nonrecursive Datalog with negation can express every first-order query, and every first-order query $\varphi(\bar x, \bar y)$ with $|\bar x|=|\bar y|=m$ can be turned into a graph whose nodes are given by $m$-tuples with edges from $\bar a$ to $\bar b$ iff $\varphi(\bar a, \bar b)$ holds. Since every NLOGSPACE problem is first-order reducible to graph reachability, to express it in $\hashd$ we then simply use first-order power of datalog to create a graph as above and then use an RPQV on it to check for reachability. 
\end{proof}


\subsection{Extension: Aggregation}\label{app:aggregation}

It is easy to extend the formal semantics of $\hashd$ with aggregation functions. The most straightforward extension is aggregation on lists, allowing atomic statements such as length$(z) = x$ for a list variable $z$. In this case, the variable $z$ should be guarded, i.e., provided to us by an RPQV. Other standard list aggregates available in GQL and SQL/PGQ can be added analogously. Using this addition, it becomes possible to write more interesting transformations. The following example (similar to the one in the body of the paper) transforms every path from Megan to Mike in the original graph into a single edge and adds a property ``length'' to it, in which it puts the length of the respective path.

$
\begin{array}{rl}
\mathit{Megan}(x) & \gets (x)\; (y)\langle y.\mathit{owner} = ``\text{Megan}"\rangle \; (x)\\
\mathit{Mike}(x) & \gets (x)\;(y)\langle y.\mathit{owner} = ``\text{Mike}"\rangle\; (x)\\
\mathit{Path}(x, y, z) & \gets 
(x)\ (\_)([z](\_)[z'] \langle z.\mathit{value} < z'.\mathit{value} \rangle )^*[z](\_)\  (y)
\end{array}
$

\medskip
$
\begin{array}{rl}
\node(x) & \gets  \#\mathit{Megan}(y, x) \lor \#\mathit{Mike}(y, x)\\
\edge(p) & \gets  \#\mathit{Path}(x, y, z, p)\\
\src(p, y) & \gets  \edge(p), \node(y), \#\mathit{Path}(z_1, z_2, z_3, p), \#\mathit{Megan}(z_1, y)\\
\tgt(p, y) & \gets  \edge(p), \node(y), \#\mathit{Path}(z1, z2, z3, p), \#\mathit{Mike}(z_3, y)\\
\prop(x, y, z) & \gets \edge(x), y = ``\text{length}", \#B(z_1, z_2, z_3, x), z = \text{length}(z_2)
\end{array}
$

%% file: gql.tex
We describe the key ingredients of the proposal that will be communicated to the ISO working groups for 
GQL and SQL.
We propose two separate additions that are
backward compatible: we thus do not propose to change the existing
behavior, \emph{which is a must for the ISO committee}. The first change concerns the behavior of \emph{patterns}, by
incorporating different behaviors of RPQs and different treatment of
variables that can be bound to single elements and/or lists.
The second change incorporates some of the features provided by
$\hashd$.

\subsection{Additional Pattern Flexibility}

In terms of the language design, we start with the following basic
principles:
\begin{enumerate}
\item more symmetry: paths need not start and end with a node;
\item all variables in a pattern can be list
    variables, except \emph{boundary variables};
\item everything written in the current GQL/SQL syntax should work as
      before.
\end{enumerate}
The current standards desire paths to be node-to-node. They ensure this with an {\em automatic node insertion} and {\em node collapse} policy in
patterns. 
For instance, a
subpattern \gqlcode{-[:a]-> -[:b]->} 
is rewritten to \gqlcode{()-[:a]->()-[:b]->()} that inserts an anonymous node between edges, and
adds end-nodes so that paths that match the subpattern start and finish with nodes. Further, 
if we have two 
consecutive nodes in a pattern, they collapse into one.
For example, in a 
pattern \gqlcode{(x:a) (y:b)}, the two nodes collapse; that is, 
both \gqlcode{x} and \gqlcode{y} must be bound to the same node, that
must have both labels \gqlcode{a} and \gqlcode{b}, just as in RPQVs.

The patterns that match paths that not necessarily start and end in a
node can be introduced in GQL by extending its \emph{matching modes}. We propose a mode 
\begin{gql*}
MATCH COLLAPSE <pattern>
\end{gql*}
in which a pattern must be fully specified, that is, no graph
elements are automatically inserted. Introducing this new mode makes our proposal backward compatible.
In the new mode, the pattern
\begin{gql*}
MATCH COLLAPSE (:Start)  
        ( (-[x]->()-[y]-> WHERE x.prop < y.prop)+ | -> )  (:End)
\end{gql*}
will match paths for which the property \gqlcode{prop} of edges increases on
paths from \verb|Start| to \verb|End| nodes. Notice that we explicitly write \gqlcode{-[x]->()-[y]->} instead of the usual \gqlcode{-[x]-> -[y]->}, since in this mode no insertions of
anonymous patterns happen anymore. In the \kw{COLLAPSE} mode,    \gqlcode{x} of the previous iteration and \gqlcode{y} of the next iteration will be matched to the same edge, as in RPQVs. 

List variables in RPQVs can also be handled by a simple syntactic extension: a new keyword
(say, \kw{ACCUMULATE}) before a path pattern expression indicates that all its internal variables should be treated as list variables (aka group variables).

\subsection{Achieving Compositionality in GQL}

Our proposal for language enhancement is based on GQL's idea ---
borrowed from Cypher --- of linear of pipelined evaluation, see
\cite{sigmod22,icdt23} and a theoretical model in
\cite{vldb25}. We outline its key
ideas next. A GQL query is a sequence of \emph{clauses}, and the mechanism of passing
information between them is called a \emph{binding
table}. That is, a clause $C$ takes two inputs: a graph and a table. The graph, however, \emph{is always the input graph $G$} and it is only the table that evolves. Thus, a GQL sequence of clauses $C_1 \  C_2 \ \cdots \ C_n$ produces 
$$ C_n\Big(\ G, \ \cdots\  C_2\big(\ G, C_1(G,T_{()})\ \big) \cdots
\Big)\;,$$
where $T_{()}$ is the table that contains a single empty tuple. Notice that the end-result is a \emph{table}.

In this framework, functionalities of $\hashd$ can be incorporated by
new clauses that {\em modify} the graph $G$ itself. 
That is, each clause now maps a pair $(G,T)$ of a graph and a
table into a new pair $(G',T')$, with the semantics of a query being
the composition of the clauses viewed as such functions. We shall add
two new clauses: \kw{CREATE NODE} and \kw{CREATE EDGE}, and
the ability to refer to values that produced new nodes and edges, as
$\#$-predicates do in $\hashd$.

Values used for generating new nodes and 
edges can come from the binding table, or from a
match, leading to  the following syntax proposal for adding
nodes:
\begin{gql*}
CREATE NODE (:<label>,...,:<label>,
             prop_name:<expr>,...,prop_name:<expr> )
FROM MATCH <pattern> | ROW
WHERE <condition>
\end{gql*}
If \kw{FROM} is followed by  \kw{MATCH}, the pattern
matching statement is performed, and every tuple in the resulting
match that satisfies \gqlcode{condition} in \kw{WHERE} gives rise to a new node. In the case of \kw{ROW}, 
every row in the binding that satisfies \gqlcode{condition} generates a new node. A node can have zero or more
labels, and zero or more properties, given by expressions that can
refer to values from either the binding table row or the result of
pattern matching. 
For example, the following clause turns every edge with
label \gqlcode{transfer} (whose \gqlcode{ts} property is in 2026 or later) into a node, keeping its \gqlcode{ts} property but renaming it to \gqlcode{timestamp}: 
\begin{gql*}
CREATE NODE (:transfer_node, timestamp:e.ts) 
FROM MATCH -[e:transfer]-> WHERE e.ts >= '2026-01-01'
\end{gql*}

\smallskip
\noindent
Next,  we need the ability to refer to properties of
graph elements that gave rise to new nodes or edges. Above, a new \gqlcode{transfer} node was created from a transfer edge
\gqlcode{e}. We propose a syntactic device \gqlcode{@e} as a way of
referencing this edge. For example, if transfer edges have a
property \gqlcode{amount} and we want to find transfer nodes that
come from edges with the amount at least 100, we would write
\begin{gql*}
MATCH (n:transfer_node) WHERE n@e.amount > 100
\end{gql*}
It is feasible to maintain this connection in the same way as it is feasible to maintain $\#$-predicates in $\hashd$, see Remark \ref{remark:hash-tables}.



Finally, we explain how edges can be added. The general pattern is the
same as for nodes but with a few differences:
\begin{gql*}
CREATE EDGE (<expr>)
 -[:<label>,...,:<label>,prop_name:<expr>,...,prop_name:<expr>]-> 
            (<expr>)
FROM MATCH <pattern> | ROW
WHERE <condition>
\end{gql*}
An edge is created for each match of a pattern or each row
in the binding table satisfying  \gqlcode{condition}, and can have zero or more
labels or properties. It needs to have its source and destination, which
are given by two expressions. If these expressions are {\em not} evaluated to  nodes, then a fresh node would be created in the
place of a source/destination. We need one additional
syntactic device: the ability to reference the source and
destination of an existing edge (see the example below). This is
needed since in the GQL standard, this check is performed by
conditions \kw{IS SOURCE | DESTINATION OF} rather than functions. 

As an example, we connect two transfer nodes if, as edges, the
destination of one was the source of the other. We illustrate this using the \kw{FROM ROW} mode: 
\begin{gql*}
MATCH (n1:transfer_node), (n2:transfer_node)
CREATE EDGE (n1) -[:new_edge]-> (n2) FROM ROW
WHERE n1@e.destination = n2@e.source
\end{gql*}
Once a clause creates nodes or edges, these are available to
subsequent clauses. For example, the \emph{increasing values in edges} query can now be expressed {\em without} using the \kw{COLLAPSE} mode as 
\begin{gql*}
MATCH (x) ( (n1:transfer_node) -[:new_edge]->(n2.transfer_node) 
         WHERE n1.ts < n2.ts )* (y)
RETURN x@e.source, y@e.destination         \end{gql*}
At the end of the query, the newly created nodes and edges can disappear by default or some of them can be made
persistent. No additional 
syntax is necessary, as one can just adapt the
existing \kw{INSERT} clause for this purpose.

Finally, note that the two syntactic devices proposed here {\em independently} fill the previously known expressivity gaps of GQL. This gives more flexibility in terms of advocating for GQL extensions, as the adoption of either would contribute to increased expressiveness.